\documentclass[11pt, twoside, letterpaper]{amsart}
\usepackage{amsmath, amssymb, hyperref, color}
\usepackage{srcltx}

\usepackage{geometry}
\newtheorem{thm}{Theorem}[section]

\newtheorem{lem}[thm]{Lemma}

\theoremstyle{definition}

\newtheorem{as}[thm]{Assumption}

\theoremstyle{remark}
\newtheorem{rem}[thm]{Remark}

\numberwithin{equation}{section}
\newcommand{\A}{\mathcal{A}}
\newcommand{\Z}{\mathcal{Z}}
\newcommand{\X}{\mathcal{X}}
\newcommand{\Y}{\mathcal{Y}}

\newcommand{\R}{\mathbb{R}}

\newcommand{\expec}{\mathbb{E}}

\newcommand{\PP}{\mathbb{P}}

\newcommand{\F}{\mathcal{F}}

\newcommand{\ud}{\mathrm d}

\newcommand{\cO}{\mathcal{O}}

\newcommand{\dbracc}[1]{[\kern-0.15em[ #1 ]\kern-0.15em]}
\newcommand{\dbraco}[1]{[\kern-0.15em[ #1 [\kern-0.15em[}
\newcommand{\dbraoc}[1]{]\kern-0.15em] #1 ]\kern-0.15em]}
\newcommand{\dbraoo}[1]{]\kern-0.15em] #1 [\kern-0.15em[}
\newcommand{\be}{\begin{equation}}
\newcommand{\ee}{\end{equation}}
\newcommand{\ba}{\begin{aligned}}
\newcommand{\ea}{\end{aligned}}
\newcommand{\ind}{\mathbb{I}}

\DeclareMathOperator{\cl}{cl}

\begin{document}
\title[Optimal consumption of multiple goods]{Optimal consumption of multiple goods in incomplete markets}%

\author{Oleksii Mostovyi}

\address{Oleskii Mostovyi, Department of Mathematics, University of Connecticut (USA)}%
\email{oleksii.mostovyi@uconn.edu}%
\thanks{
The author would like to thank Robert C. Merton for suggesting this problem and for a discussion on the subject of the paper. The author is also thankful to an anonymous referee for valuable comments. The author's research is supported by NSF Grant DMS-1600307. Any opinions, findings and conclusions or recommendations expressed in this material are those of the author and do not necessarily reflect those of the National Science Foundation.}
\date{\today}%

\begin{abstract}
We consider the problem of optimal consumption of multiple goods in incomplete semimartingale markets. We formulate the dual problem and identify conditions that allow for existence and uniqueness of the solution and give a characterization of the optimal consumption strategy in terms of the dual optimizer.
 We illustrate our results with examples in both complete and incomplete models. In particular, we construct closed-form solutions in some incomplete models.
\end{abstract}

\maketitle

\section{Introduction}
The problem of optimal consumption of multiple goods has been investigated in \cite{Fisher75, Breeden79}.
For a single consumption good in continuous-time settings, it was first formulated  in \cite{Merton69}. Since then, this problem was analyzed in a large number of papers in both complete and incomplete settings with a range of techniques based on Hamilton-Jacobi-Bellman equations, backward stochastic differential equations, and convex duality being used for its analysis.

In the present paper, we formulate a problem of optimal consumption of multiple goods in a general incomplete semimartingale model of a financial market. We construct the dual problem and characterize optimal consumption policies in terms of the solution to the dual problem. We also identify mathematical conditions, that allow for existence and uniqueness of the solution and a dual characterization. We illustrate our results by examples, where in particular we obtain closed-form solutions in incomplete markets. Our proofs rely on certain results on weakly measurable correspondences for Carath\'eodory functions, multidimensional convex-analytic techniques, and some recent advances in stochastic analysis in mathematical finance, in particular,  the characterization of the ``no unbounded profit with bounded risk'' condition in terms of non-emptiness of the set of equivalent local martingale deflators from \cite{MostovyiNUPBR, KabanovKardarasSong} and sharp conditions for solvability of the expected utility maximization problem in a single good setting from \cite{Mostovyi2015}. 

The remainder of this paper is organized as follows: in Section \ref{sec:2} we specify the model setting, formulate the problem, and state main results (in Theorem \ref{mainTheorem}). In Section \ref{Examples} we discuss various specific cases. In particular, we present there the structure of the solution in complete models and the additive utility case as well as closed-form solutions in some incomplete models (with and without an additive structure of the utility). We conclude the paper with Section \ref{proofs}, which contains proofs.

\section{Setting and main results}	\label{sec:2}

\subsection{Setting}	\label{sec:setting}

Let $\widetilde S=(\widetilde S_t)_{t\geq0}$ an $\R^d$-valued semimartingale, representing the discounted prices\footnote{Since we allow preferences to be stochastic (see the definition below), there is no loss of generality in assuming that asset prices are discounted, see \cite[Remark 2.2]{Mostovyi2015} for a more detailed explanation of this observation.} of $d$ risky assets on a complete stochastic basis $(\Omega,\F,(\F_t)_{t\in[0,\infty)},\PP)$, with $\F_0$ being 
the trivial $\sigma$-algebra. 
We fix a \emph{stochastic clock} $\kappa=(\kappa_t)_{t\geq0}$, which is a nondecreasing, c\`adl\`ag, adapted process, such that
\be	\label{clock}
\kappa_0 = 0,
\qquad
\PP(\kappa_{\infty}>0)>0
\qquad\text{and}\qquad
\kappa_{\infty}\leq \bar A,
\ee
where $\bar A$ is a positive constant.
The stochastic clock $\kappa$ specifies times when consumption is assumed to occur. 
Various optimal investment-consumption problems can be recovered from the present general setting by suitably specifying the clock process $\kappa$. For example, the problem of maximizing expected utility  of terminal wealth at some finite investment horizon $T<\infty$ can be recovered by simply letting $\kappa\triangleq\ind_{\dbraco{T,\infty}}$. Likewise, maximization of expected utility from consumption only up to a finite horizon $T<\infty$ can be obtained by letting $\kappa_t\triangleq \min(t, T)$, for $t\geq0$.
Other specifications include maximization of utility form lifetime consumption, from consumption at a finite set of stopping times, and from terminal wealth at a random horizon, see e.g.,\cite[Examples 2.5-2.9]{Mostovyi2015} for a description of possible standard choices of the clock process~$\kappa$.


We suppose that there are $m$ different consumption goods, where $S^k_t$ denotes the discounted price of commodity $k$ at time $t$. We assume that for each $k\in\{1,\dots, m\}$, $S^k=(S^k_t)_{t\geq0}$ is a {\it strictly} positive optional processes
 on $(\Omega,\F,(\F_t)_{t\in[0,\infty)},\PP)$.

A \emph{portfolio} is defined by a triplet $\Pi=(x,H,c)$, where $x\in\R$ represents an initial capital, $H=(H_t)_{t\geq0}$ is a $d$-dimensional $\widetilde S$-integrable process, 
$H^j_t$ represents the holdings in the $j$-th risky asset at time $t$, $j \in\{1,\dots,d\}$, $t\geq 0$, $c$ is an $m$-dimensional {\it consumption process}, whose every component $(c^k_t)_{t\geq0}$ is a nonnegative optional process representing the consumption {\it rate} of commodity $k$, $k = \{1,\dots, m\}$. 
The 
 wealth process $X=(X_t)_{t\geq0}$ of a portfolio $\Pi=(x,H,c)$ is defined as
\be\label{X}
X_t \triangleq x + \int_0^tH_u\,\ud \widetilde S_u - \sum\limits_{k=1}^m\int_0^tc^k_uS^k_u\,\ud \kappa_u,
\qquad t\geq0.
\ee
\subsection{Absence of arbitrage}	\label{sec:NUPBR}
The main objective of this part is to specify the no-arbitrage type condition \eqref{NUPBR} below. As it is commonly done in the literature (see for example \cite{KS99}), we begin defining $\X$ to be the collection of all nonnegative wealth processes associated to portfolios of the form $\Pi=(1,H,0)$,~i.e.,
\[
\X \triangleq \left\{X\geq0 : X_t=
1 + \int_0^tH_u\ud \widetilde S_u,\quad t\geq0\right\}.
\]

In this paper, we suppose the following no-arbitrage-type condition:
\be	\tag{NUPBR}	\label{NUPBR}
\text{the set } 
\X_T \triangleq \bigl\{X_T : X\in\mathcal{X}\bigr\} 
\text{ is bounded in probability, for every }T\in\R_+,
\ee
where (NUPBR) stands for {\it no unbounded profit with bounded risk}. This condition was originally introduced in 
\cite{KaratzasKardaras07}. It is proven in \cite[Proposition 1]{Kardaras10}, that \eqref{NUPBR} is equivalent to another (weak) no-arbitrage condition, namely absence of \emph{arbitrages of the first kind} on $[0,T]$, see \cite[Definition 1]{Kar14}. 

A useful characterization of \eqref{NUPBR} is given via the set of  \emph{equivalent local martingale deflators (ELMD)} that is defined as follows:
\begin{equation}\label{setZ}
\begin{array}{rl}
\mathcal{Z} \triangleq \bigl\{ Z >0 :\;& Z \text{ is a c\`adl\`ag local martingale such that  } Z_0 =1 \text{ and } \\
& ZX = (Z_tX_t)_{t\geq 0} \text{ is a local martingale for every } X \in \mathcal{X} \bigr\}.
\end{array}
\end{equation}

It is proven in \cite[Proposition 2.1]{MostovyiNUPBR} (see also \cite{KabanovKardarasSong}) that condition \eqref{NUPBR} holds if and only if $\Z\neq\emptyset$.
This result was previously established in the one-dimensional case in the finite time horizon in \cite[Theorem 2.1]{Kardaras12}. Also, \cite[Theorem 2.6]{TS14} contains a closely related result (in a finite time horizon) in terms of {\it strict $\sigma$-martingale densities}, see \cite{TS14} for the corresponding definition and details.


\begin{rem}	\label{rem:comp_mostovyi}
Condition \eqref{NUPBR} is weaker than the existence of an equivalent martingale measure (see for example \cite[p. 463]{DS94} for the definition an equivalent martingale measure), another classical no-arbitrage type assumption, which in the infinite time horizon is  even stronger than
\begin{equation}\label{NFLVR}
\{Z\in\Z : Z\text{ is a martingale}\}\neq\emptyset.
\end{equation}
Note that in the {\it finite time horizon} setting, \eqref{NFLVR} is equivalent to the existence of an equivalent martingale measure. Besides, \eqref{NFLVR} is apparently stronger than \eqref{NUPBR} (by comparison of \eqref{setZ} and \eqref{NFLVR} combined with \cite[Proposition 2.1]{MostovyiNUPBR}).  We also would like to point out that \eqref{NFLVR} holds in every original formulation of \cite{Merton69}, where the problem of optimal consumption from investment (in a single consumption good setting) was introduced, including the infinite-time horizon case. In general, \eqref{NFLVR} can be stronger than \eqref{NUPBR}.  A classical example, where \eqref{NUPBR} holds but \eqref{NFLVR} fails, corresponds to the three-dimensional Bessel process driving the stock price, see e.g., \cite[Example 4.6]{KaratzasKardaras07}.
\end{rem}
\subsection{Admissible consumptions}
For a given initial capital $x>0$, an $m$-dimensional optional consumption process $c$ is said to be \emph{$x$-admissible} if there exists an $\R^d$-valued predictable $\widetilde S$-integrable process $H$ such that the wealth process $X$ in \eqref{X}, corresponding to the portfolio $\Pi=(x,H,c)$ is nonnegative; the set of $x$-admissible consumption processes corresponding to a stochastic clock $\kappa$ is denoted by $\A(x)$.
For brevity, we denote $\A\triangleq\A(1)$.

\subsection{Preferences of a rational economic agent}	\label{sec:invest}
Building from the formulation of \cite{MertonCTF}, we assume that preferences of a rational economic agent are represented by a \emph{optional utility-valued process} (or simply a {\it utility process}) $U=U(t,\omega,x):[0,\infty)\times\Omega\times[0,\infty)^m\rightarrow\R\cup\{-\infty\}$, where for every $(t, \omega)\in [0,\infty)\times \Omega$,  $U(t,\omega, \cdot)$ is an Inada-type utility function, i.e., $U(t,\omega, \cdot)$ satisfies the following (technical) assumption.

\begin{as}	\label{as:U}
For every $(t,\omega)\in[0,\infty)\times\Omega$, the function $$\R^m_{+} \ni x\mapsto U(t,\omega,x)\in \R\cup\{-\infty\}$$ is strictly concave, strictly increasing in every component,  finite-valued and continuously differentiable in the interior of the positive orthant, and satisfies the Inada conditions
\[
\underset{x_i\downarrow0}{\lim} \, {\partial_{x_i}U}(t,\omega,x) = \infty
\quad\text{and}\quad
\underset{x_i\uparrow\infty}{\lim} \, {\partial_{x_i}U}(t,\omega,x) = 0,\quad i = 1,\dots, m,
\]
where ${\partial_{x_i}U}(t,\omega, \cdot):\R^m_{++} \mapsto\R$ is the partial derivative of $U(t,\omega,\cdot)$ with respect to the $i$-th spatial variable
\footnote{For the results below, we only need to specify the gradient of $U(t,\omega, \cdot)$ in the {\it interior} of the first orthant, i.e., at the points $x\in\R^m$, where $U(t,\omega, x)$ is (finite-valued and) differentiable.}. On the boundary of the first orthant, by upper semicontinuity, we {suppose} that $U(t,\omega,x)=\limsup\limits_{x'\to x}U(t,\omega,x')$ (note that some of these values may be $-\infty$ and that $U(t,\omega,x) = \lim\limits_{t\downarrow 0}U(t,\omega,x + t(x'-x))$, where $x'$ is an arbitrary element in the interior of the first orthant, see \cite[Proposition B.1.2.5]{LH04}). Finally, for every $x\in \R^m_{+}$, we assume that the stochastic process $U(\cdot,\cdot,x)$ is optional.
\end{as}

\begin{rem} The Inada conditions in Assumption \ref{as:U} were introduced in \cite{Inada63}. These are technical assumptions that have natural economic interpretations and that allow for a deeper tractability of the problem (as e.g., in \cite{KS99}). 
Likewise, the semicontinuity of $U$ is imposed for regularity purposes. It also used in e.g., \cite{Pietro1, Pietro2}.
 \end{rem}
In particular, modeling preferences via utility process allows to take into account utility maximization problems under a change of num\'eraire (see e.g., \cite[Example 4.2]{MostovyiRE}). This is the primary reason why we suppose that the prices of the traded stocks are discounted, as this allows to simplify notations without any loss of generality.
Note also that Assumption \ref{as:U} does not make any requirement on the {\it asymptotic elasticity} of $U$, introduced in \cite{KS99}.

To a utility process $U$ satisfying Assumption \ref{as:U}, we  associate the \emph{primal value function}, defined as
\be	\label{primalProblem}
u(x) \triangleq \sup_{{\mathbf c} = (c^1,\dots,c^m)\in\A(x)}\expec\left[\int_0^{\infty}U(t,\omega,{\mathbf c}_t)\,\ud\kappa_t\right],\quad x>0.
\ee
To ensure that the integral above is well-defined, we use the convention 
\be\label{9171}
\expec\left[\int_0^{\infty}U(t,\omega,{\mathbf c}_t)\,\ud\kappa_t\right]\triangleq-\infty\quad  if\quad \expec\left[\int_0^{\infty}U^-(t,\omega,{\mathbf c}_t)\,\ud\kappa_t\right]=\infty,
\ee
where $U^-(t,\omega,\cdot)$ is the negative part of $U(t,\omega,\cdot)$. Note that formulation \eqref{primalProblem} is a generalization  of the formulation in \cite[p. 205]{MertonCTF}, in the form \eqref{primalProblem}  we allow for stochastic preferences and include several standard formulations as particular cases.

\subsection{Dual problem}	\label{sec:dualProblem}
In order to specify model assumptions that ensure existence and uniqueness of solutions to \eqref{primalProblem} and to give a characterization of this solution, we need to formulate an appropriate dual problem. 
Let us define
\be\label{def:U*}
U^{*}(t,\omega, x)\triangleq \sup\limits_{\substack{(x_1,\dots,x_m)\in\R^m_{+}:\\\sum\limits_{k = 1}^mS^k_t(\omega)x_k \leq x}}U\left(t,\omega,x^1,\dots,x^m\right),\quad (t,\omega,x)\in [0,\infty)\times \Omega\times[0,\infty).
\ee

Let us set a family of transformations $A:[0,\infty)\times\Omega\times\R^m\mapsto \R$, as $$A (t,\omega, x_1,\dots,x_m)\triangleq S^1_t(\omega)x_1 + \dots + S^m_t(\omega)x_m,\quad (t,\omega,x_1,\dots,x_m)\in[0,\infty)\times\Omega\times[0,\infty)^m.$$
Note that for every $(t,\omega)\in[0,\infty)\times \Omega$, $A(t,\omega,\cdot)$ is a linear transformation from $\R^m$ to $\R$ and $U^{*}(t,\omega,\cdot)$ is the image of $U(t,\omega,\cdot)$ under $A(t,\omega,\cdot)$ (see e.g., \cite[p. 96]{LH04} for the definition and properties of the {\it image of a  function under a linear mapping}\footnote{Equivalently, see \cite[Theorem 5.2]{Rok}, where $U^*(t,\omega, \cdot)$ is named the image of $U(t,\omega,\cdot)$ under the linear transformation $A(t,\omega,\cdot)$
, $(t,\omega)\in[0,\infty)\times\Omega$.}). 
We define a stochastic field $V^{*}$ as the pointwise conjugate of $U^{*}$ (equivalently, as the pointwise conjugate of the image function of $U$ under $A$) in the sense that
\[
V^{*}(t,\omega,y) \triangleq \sup_{x>0}\left(U^*(t,\omega,x)-xy\right),
\qquad (t,\omega,y)\in[0,\infty)\times\Omega\times[0,\infty),
\]
where $\sup\limits_{x>0}$ and $\sup\limits_{x\geq 0}$ coincide thanks to continuity of $U^*$ established in Lemma \ref{lem:U*}.
We also introduce the following set of dual processes:
\begin{align*}
\Y(y) \triangleq \cl\bigl\{Y :\;& Y\text{ is c\`adl\`ag adapted and }\\
&
0\leq Y\leq yZ \text{ $(\ud\kappa\times\PP)$-a.e. for some }Z\in\Z\bigr\},
\end{align*}
where the closure is taken in the topology of convergence in measure $(\ud\kappa\times\PP)$ on the measure space of real-valued optional processes $\left(\Omega\times [0,\infty), \cO,\ud\kappa\times\PP\right)$, where $\cO$ is the optional sigma-field.  We write $\Y\triangleq\Y(1)$ for brevity. Note that $\mathcal Y$ is closely related to - but different from - the set with the same name in \cite{KS99}.
The value function of the dual optimization problem, or equivalently, the \emph{dual value function}, is then defined as
\be	\label{dualProblem}
v(y) \triangleq \inf_{Y\in\Y(y)}\expec\left[\int_0^{\infty}V^{*}(t,\omega,Y_t)\,\ud\kappa_t\right]{,\quad y>0},
\ee
with the convention $\expec[\int_0^{\infty}V^{*}(t,\omega,Y_t)\,\ud\kappa_t]\triangleq\infty$ if $\expec[\int_0^{\infty}{V^{*}}^+(t,\omega,Y_t)\,\ud\kappa_t]=\infty$, where ${V^{*}}^+(t,\omega,\cdot)$ is the positive part of ${V^{*}}(t,\omega,\cdot)$.
We are now in a position to state the following theorem, 
which is the main result of this paper.
\begin{thm}	\label{mainTheorem}
Assume that conditions \eqref{clock} and \eqref{NUPBR} hold true and let $U$ satisfies Assumption \ref{as:U}. Let us also suppose that
\be	\label{eq:finiteness}
v(y)<\infty\quad\text{ for every }y>0
\quad\text{and}\quad
u(x)>-\infty\quad\text{ for every }x>0.
\ee
Then we have 
\begin{enumerate}
\item[(i)]
$u(x)<\infty$, for every $x>0$, and $v(y)>-\infty$, for every $y>0$, i.e., the value functions are \underline{finite-valued}.
\item[(ii)]
The functions $u$ and $-v$ are continuously differentiable on $(0,\infty)$, strictly concave, strictly increasing and satisfy the Inada conditions
\begin{equation}\label{453}
\begin{array}{rclccrcl}
\underset{x\downarrow0}{\lim} \, u'(x) & = &\infty, &\quad&
\underset{y\downarrow0}{\lim} \, -v'(y) &=& \infty,	\\
\underset{x\rightarrow\infty}{\lim} \, u'(x) &=& 0, &\quad&
\underset{y\rightarrow\infty}{\lim} \, -v'(y) &=& 0. \\
\end{array}
\end{equation}
\item[(iii)]
For every $x>0$ and $y>0$, the solutions $\widehat{c}(x)=(\widehat{c}^1(x),\dots, \widehat{c}^m(x))$ to \eqref{primalProblem} and $\widehat{Y}(y)$ to \eqref{dualProblem} exist and are unique and, if $y=u'(x)$, we have the optimality characterizations
\be\label{eq:3212}
\widehat Y_t(y)(\omega)=\frac{
{\partial_{x_i}U}\bigl(t,\omega, \widehat c^1_t(x)(\omega),\dots,\widehat c^m_t(x)(\omega)\bigr)}{S^i_t(x)(\omega)},\quad (\ud\kappa\times \PP)\text{-a.e.},\quad i = 1,\dots, m.
\ee
and 
\be\label{eq:3221}
\hat{Y}_t(y)(\omega) = U^*_x\bigl(t,\omega,\sum\limits_{i = 1}^m\hat{c}^i_t(x)(\omega) S^i_t(\omega) \bigr),
\qquad (\ud\kappa\times\PP)\text{-a.e.},
\ee
with ${U_x^{*}}$ denoting the partial derivative of $U^*$ with respect to its third argument.
\item[(iv)]
For every $x>0$, the constraint $x$ is binding in the sense that
\be\label{451}
\mathbb{E}\left[\int_0^{\infty}\sum\limits_{i = 1}^m\widehat{c}^i_t(x)S^i_t\dfrac{\widehat{Y}_t(y)}{y}\,\ud\kappa_t\right] = x,\quad where~y=u'(x).
\ee
\item[(v)] 
 The functions $u$ and $v$ are Legendre conjugate, i.e., 
\begin{equation}\label{452}
v(y) = \underset{x>0}{\sup} \bigl(u(x)-xy\bigr),\quad y>0, \qquad
u(x) = \underset{y>0}{\inf} \bigl(v(y)+xy\bigr),\quad x>0.
\end{equation}

\item[(vi)]
The dual value function $v$ can be represented as
\be	\label{eq:v_defl}
v(y) = \inf_{Z\in\mathcal Z}\expec\left[\int_0^{\infty}V(t,\omega,yZ_t(\omega))\,\ud\kappa_t(\omega)\right],\quad y>0.
\ee

\end{enumerate}

\end{thm}
\begin{rem}[On sufficient conditions for the validity of \eqref{eq:finiteness}]
Condition \eqref{eq:finiteness} holds if there exists one primal element $c\in\mathcal A$ and one dual element $Y\in\mathcal Y$ such that 
$$\mathbb E\left[\int_0^\infty U\left(t,\omega,zc^1_t,\dots,zc^m_t\right)\ud\kappa_t\right]>-\infty\quad and\quad \mathbb E\left[\int_0^\infty V^{*}\left(t,\omega,zY_t\right)\ud\kappa_t\right]<\infty,\quad z>0.$$
In particular, 
for every $x>0$, as an $m$-dimensional optional process with constant values $\left(\tfrac{x}{\bar Am},\dots,\tfrac{x}{\bar Am}\right)$ belongs to $\mathcal A(x)$, a sufficient condition in \eqref{eq:finiteness} for the finiteness of $u$ is
$$\mathbb E\left[\int_0^\infty U\left(t,\omega,\tfrac{x}{\bar Am},\dots,\tfrac{x}{\bar Am}\right)\ud\kappa_t\right]>-\infty,\quad x>0,$$
which typically holds if $U$ is nonrandom. Likewise, as $\Z\neq\emptyset$ (by \eqref{NUPBR} and \cite[Proposition 2.1]{MostovyiNUPBR}), finiteness of $v$ holds if for one equivalent local martingale deflator $Z$, we have
$$ \mathbb E\left[\int_0^\infty V^{*}\left(t,\omega,yZ_t\right)\ud\kappa_t\right]<\infty,\quad y>0.$$
\end{rem}

\section{Examples}\label{Examples}
\section*{Complete market solution and dual characterization}
If the model is complete, the dual characterization of the optimal consumption policies has a particularly nice form,
as $\Z$ contains a unique element, $Z$. The solutions corresponding to different $y$'s in the dual problem \eqref{dualProblem} are $yZ$, $y>0$. Therefore, in \eqref{eq:3221} and \eqref{eq:3212} we have $\widehat Y(y)= yZ$, $y>0$.

\section*{Special case: Additive utility}

An important example of $U^{*}$ corresponds to $U$ having an additive form with respect to its spatial components, i.e., when
$$U(t,\omega, c_1,\dots, c_m) = U^1(t,\omega,c_1)+\dots+U^m(t,\omega, c_m),\quad (t,\omega)\in[0,\infty)\times \Omega,$$
where for every $k=1,\dots,m$, $U^k$ is a utility process in the sense of \cite[Assumption 2.1]{Mostovyi2015} and a utility process in sense of the Assumption \ref{as:U} with $m=1$. In this case, for every $(t,\omega)\in[0,\infty)\times\Omega$, $U^*(t,\omega,\cdot)$ is given by the {\it infimal convolution} of $U^k(t,\omega,\cdot)$'s, see the definition in e.g., \cite[p. 34]{Rok}.  Let $V^i(t,\omega,\cdot)$ denote the convex conjugate of $U^i(t,\omega,\cdot)$, $i=1,\dots,m$. Then the convex conjugate of $U^{*}(t,\omega,\cdot)$ is $V^{*}(t,\omega, \cdot)$ given by 
$$V^*(t,\omega,\cdot) = V^1(t,\omega,\cdot) + \dots +V^m(t,\omega,\cdot).$$
 This result was established  e.g., in \cite[Theorem 16.4, p. 145]{Rok}. In this case, the optimal $\widehat c(x) = (\widehat c^1(x),\dots,\widehat c^m(x))$ has a more explicit characterization via $I_i(t,\omega,\cdot) \triangleq \left(U^i_x\right)^{-1}(t,\omega,\cdot)$, the the pointwise inverse of the partial derivative of $U^i(t,\omega,\cdot)$ with respect to the third argument, as \eqref{eq:3212} can be solved for $\widehat c^i(x)$, $i = 1,\dots,m$, as follows
\be\label{eq:3261}
\widehat c^i_t(x)(\omega)= I_i\left(t,\omega,\widehat Y_t(y)(\omega)S^i_t(\omega)\right),\quad (\ud\kappa\times \PP)\text{-a.e.},\quad i = 1,\dots, m.
\ee
Using \eqref{eq:3221}, we can restate \eqref{eq:3261} as
\be\nonumber
\widehat c^i_t(x)(\omega)= I_i\left(t,\omega,U^*_x\bigl(t,\omega,\widehat c^{*}_t(x)(\omega) \bigr)S^i_t(\omega)\right),\quad (\ud\kappa\times \PP)\text{-a.e.},\quad i = 1,\dots, m,
\ee
where $\widehat c^{*}(x)$ is the optimizer to the auxiliary problem \eqref{eq:u*} corresponding to the initial wealth $x>0$. 

\begin{rem}
In the following three examples we consider some incomplete models that admit closed-form solutions for one good and show how these results apply to multiple good settings.  
\end{rem}

\section*{Example of a closed form solution in an incomplete model with additive logarithmic utility}
Let us suppose that $d$ traded discounted assets are modeled with Ito processes of the form
\begin{equation}\label{4161}
d\widetilde S^i_t = \widetilde S^i_t b^i_tdt + \widetilde S^i_t \sum\limits_{j=1}^n \sigma^{ij}_tdW^j_t,\quad i = 1,\dots,d,\quad \widetilde S_0\in\R^d,
\ee
where $W$ is an $\R^n$-valued standard Brownian motion and $b^i$, $\sigma^{ij}$, $i = 1,\dots,d$, $j=1,\dots, n$, are predictable processes, such that the unique strong solution to \eqref{4161} exists, see e.g., \cite{KS98}. Let us suppose that there are $m$ consumption goods and that the value function of a rational economic agent is given by 
$$\sup\limits_{c\in \mathcal A(x)}\mathbb E\left[\int_0^T e^{-\nu t}\log(c_1\dots c_m)dt\right],\quad x>0,$$
(with the same convention as the one specified after \eqref{primalProblem}),
where an impatience rate $\nu$ and a time horizon $T$ are positive constants. Note that in this case 
$\kappa_t = \frac{1 - e^{-\nu t}}{\nu}$, $t\in[0,T]$, i.e., $\kappa$ is deterministic.
Let us also suppose that there exists an $\R^d$-valued process $\gamma$, such that
$$b_t - \sigma_t\sigma_t^T \gamma_t = 0\quad (\ud \kappa\times \mathbb P)-a.e.$$
Let $\mathcal E$ denotes the Dol\'eans-Dade exponential. Then, using \cite[Theorem 3.1 and Example 4.2]{GollKallsen00} and Theorem \ref{mainTheorem}, we get
$$\widehat c^{*}_t(x) = \frac{x\nu}{1 - e^{-\nu T}}\mathcal E\left(\int_0^{\cdot}\gamma^T_sd\widetilde S_s\right)_t,\quad  x>0,$$
$$\widehat c^i_t(x) = \frac{\widehat c^{*}_t(x)}{S^i_t M},\quad i=1,\dots,m, \quad x>0,$$
$$\widehat Y_t(y) = \frac{y}{\mathcal E\left(\int_0^{\cdot}\gamma^T_s d\widetilde S_s\right)_t},\quad y>0,\quad t\in[0,T].$$



\section*{Example of a closed-form solution and dual characterization in an incomplete additive case}
Let us fix a filtered probability space $(\Omega, \mathcal F, \mathbb P)$, where $(\mathcal F_t)_{t\geq 0}$ is the augmentation of the filtration generated by a two-dimensional Brownian motion $(W^1, W^2)$. 
Let us suppose that there are two traded securities: a risk-free asset $B$, such that $$B_t  = e^{rt},\quad t>0,$$ 
where $r$ is a nonnegative constant, and a risky stock $\widetilde S$ with the dynamics 
$$d\widetilde S_t = \widetilde S_t\mu_t dt + \widetilde S_t\sigma_t dW^1_t,\quad t\geq 0,\quad\widetilde S_0 \in\R_{+},$$
where processes $\mu$ and $\sigma$ are such that $\theta_t = \frac{\mu_t - r}{\sigma_t}$, $t\geq 0$, the market price of risk process, follows the Ornstein-Uhlenbeck process
$$d\theta_t = -\lambda_{\theta}(\theta_t - \bar \theta)dt + \sigma_\theta\left(\rho dW^1_t + \sqrt{1-\rho^2} dW^2_t\right),\quad t\geq 0,\quad \theta_0\in\R_{+},$$
where $\lambda_{\theta}$, $\sigma_{\theta}$, and $\bar \theta$ are positive constants, $\rho\in(-1,1)$.
 Let us also assume that 
$\kappa$ corresponds to the expected utility maximization from terminal wealth, i.e., $\kappa =\ind_{\dbraco{T,\infty}}$, $T\in\R_{+}$, that there are $m$ consumption goods, where $S^i$, $i = 1,\dots,m$, are {\it deterministic},  and $$U(T,\omega, c_1,\dots,c_m) = \frac{c_1^p}{p} + \dots + \frac{c_m^p}{p},\quad (c_1,\dots,c_m)\in \R^m_{+},\quad \omega\in\Omega,$$
where $p<0$. 
Let us set
$$q\triangleq \frac{p}{1-p},\quad A\triangleq \sum\limits_{i=1}^m(S^i_T)^{-q},\quad and \quad 
B\triangleq A^{1-p}
.$$Then, by direct computations, we get $$U^*(T,\omega,x) = \frac{x^p}{p}B,\quad x>0.$$
Using the argument in \cite{KO96}, 
one can express the optimal trading strategy is $\widehat H(x)$ in a closed form in terms of a solution to a system of (nonlinear) ordinary differential equations (see \cite[p. 147]{KO96}), 
where $\widehat H_t(x)$ is the number of shares of the risky asset in the portfolio at time $t$, $t\in[0,T]$. With $\widehat X(x)$ such that $$d\widehat X_t(x) = \widehat H_t(x)d\widetilde S_t + (\widehat X_t(x)-\widehat H_t(x)\widetilde S_t)rdt,\quad \widehat X_0(x) = x,$$ using Theorem \ref{mainTheorem}
, we get 
$$\widehat {c}^{*}_T(x) = \widehat X_T(x),\quad x>0,$$
$$\widehat Y_T(y) = \frac{y}{\mathbb E\left[\left(\widehat {c}^{*}_T(1) \right)^p\right]}\left(\widehat {c}^{*}_T(1)\right)^{p-1},\quad y>0,$$
$$\widehat c^i_T(x) = \frac{\widehat c^{*}_T(x)}{A}(S^i_T)^{-(1+ q)},\quad x>0.$$

\section*{Example of a closed-form solution and dual characterization in an incomplete non-additive case}
Here we will suppose that
$\kappa =\ind_{\dbraco{T,\infty}}$, where $T\in \R_{+}$, and let 
$$U(t,\omega,c_1,c_2) = - \frac{c_1^{p_1}}{p_1}\frac{c_2^{p_2}}{p_2},\quad p_1<0,p_2<0,$$
i.e., there are two consumption goods.
One can see that $U(t,\omega,\cdot)$ is jointly concave, since the Hessian of $-U(t,\omega,\cdot)$ is positive definite on $\R^2_{++}$. We also extend $U(t,\omega,\cdot)$ to the boundary of $\R^2_{+}$ by $-\infty$. Then, with $p \triangleq p_1 + p_2<0$,  $U^*$ is given by
$$U^*(t,\omega, x) = \frac{x^p}{p}\frac{(-p_1)^{p_1 - 1}(-p_2)^{p_2-1}}{(-p)^{p-1}}(S^1_t)^{-p_1}(S^2_t)^{-p_2},\quad x>0.$$
Let us define $G \triangleq \frac{(-p_1)^{p_1 - 1}(-p_2)^{p_2-1}}{(-p)^{p-1}}(S^1_T)^{-p_1}(S^2_T)^{-p_2}$.
Then $U(T,\omega, x) = \frac{x^p}{p}G(\omega)$, $x>0$.
Let us suppose that $W^1$ and $W^2$ are two Brownian motions with a fixed correlation $\rho$ such that $0<|\rho|<1$. 
Let $(\mathcal F_t)_{t\geq 0}$ be the usual augmentation of the filtration generated by $W^1$ and $W^2$ and $(\mathcal G_t)_{t\geq 0}$ be the usual augmentation of the filtration generated by $W^2$.
We also assume that there is a bond $B$ and a stock $\widetilde S$ on the market. Their dynamics are given by
$$d\widetilde S_t = \widetilde S_t(\mu_t dt + \sigma_t dW^1_t ),\quad \widetilde S_0 \in\R,$$
$$dB_t = B_tr_t dt,\quad B_0 = 1,$$
where the drift $\mu$, volatility $\sigma$, and sport interest rate $r$ are bounded, progressively measurable processes with respect to $({\mathcal G}_t)_t\in[0,T]$,  and $\sigma$ is strictly positive.

Let us suppose that $S^1_T$ and $S^2_T$ are $\mathcal G_T$-measurable random variables with moments of all orders. Then $G$ is also $\mathcal G_T$-measurable random variable with moments of all orders (by H\"older's inequality) and the auxiliary value function $u^*$ defined in \eqref{eq:u*} satisfies the settings of \cite{Tehranchi04}. 
 Also, as $u^{*}(x)\geq\frac{x^p}{p}\mathbb E[G] >-\infty$ and since $V(T,\omega, \cdot)$ is negative-valued (thus, $v(y)\leq0$), the assumption \eqref{eq:finiteness} holds.

Let us set $$\lambda_t \triangleq \frac{\mu_t - r_t}{\sigma_t},\quad \delta \triangleq \frac{1 - p}{1 - p + \rho^2 p},\quad \frac{d\mathbb Q}{d\mathbb P} \triangleq \exp\left(-\frac{\rho^2 p^2}{2(1-p)^2}\int_0^T\lambda^2_sds + \frac{\rho p}{1-p}\int_0^T\lambda_sdW^2_s \right),$$
$$K_t \triangleq  \frac{p}{(1-p)}\left(\lambda_t + \rho\delta\frac{\beta_t}{\mathbb {E^Q}[\exp(\int_0^T (r_s/\delta) ds)|\mathcal F_t]}\right),\quad t\in[0,T].$$
Then, using \cite[Proposition 3.4]{Tehranchi04} and Theorem \ref{mainTheorem}, we deduce that 
$$\widehat c^{*}_T(x) = x\exp\left(\int_0^T\left(r + K_s\lambda_s -\tfrac{1}{2}K^2_s\right)ds + \int_0^T K_sdW^1_s\right),\quad x>0,$$
$$\widehat Y_T(y) =\frac{y}{\mathbb E\left[\left( \widehat c^{*}_T(1)\right)^p\right]}\exp\left(\int_0^T(p-1)\left(r + K_s\lambda_s -\tfrac{1}{2}K^2_s\right)ds + \int_0^T (p-1)K_sdW^1_s\right) ,\quad y>0,$$
$$\widehat c^i_T = \frac{\widehat c^{*}_T(x) p_i}{pS^i_T},\quad i = 1,2,\quad x>0,$$
are the optimizers to \eqref{dualProblem}, \eqref{eq:u*}, and \eqref{primalProblem}, respectively. From Theorem \ref{mainTheorem}, we conclude that for every $x>0$, $\widehat c^i_T(x)$, $i=1,2,$ and $\widehat Y_T(u'(x))$ are related via \eqref{eq:3212} and \eqref{eq:3221}.

\section{Proofs}\label{proofs}
We begin from a characterization of the utility process $U^*$ defined in \eqref{def:U*}.
\begin{lem}\label{lem:U*} Let $U$ satisfies Assumption \ref{as:U} and $U^{*}$ be defined in \eqref{def:U*}. Then, $U^{*}$ is an Inada-type utility process for $m=1$ in the sense of Assumption \ref{as:U}, i.e., $U^{*}$ satisfies:
\begin{enumerate}\item 
For every $(t,\omega)\in[0,\infty)\times\Omega$, the function $x\mapsto U^{*}(t,\omega,x)$ is  finite-valued on $(0,\infty)$, strictly concave, and strictly increasing. 

\item For every $(t,\omega)\in[0,\infty)\times\Omega$, the function $x\mapsto U^{*}(t,\omega,x)$ is continuously differentiable on $(0,\infty)$ and satisfies the Inada conditions
\[
\underset{z\downarrow0}{\lim} \, {U_x^{*}}(t,\omega,z) = \infty
\qquad\text{and}\qquad
\underset{z\uparrow \infty}{\lim} \, {U_x^{*}}(t,\omega,z) = 0.
\]
\item For every $(t,\omega)\in[0,\infty)\times\Omega$, at $z=0$, we have $$U^{*}(t,\omega,0)=\lim_{z\downarrow0}U^{*}(t,\omega,z)$$ (note that this value may be $-\infty$). 
\item For every $z\geq0$, the stochastic process $U^{*}(\cdot,\cdot,z)$ is optional.
\end{enumerate}
\end{lem}
\begin{proof} 
For every $(t,\omega)\in[0,\infty)\times\Omega$, 
as $U^{*}(t,\omega, \cdot)$ is an image function under an appropriate linear transformation of a {\it concave} function $U(t,\omega, \cdot)$, therefore using e.g., \cite[Theorem B.2.4.2]{LH04}, one can show that $U^{*}(t,\omega, \cdot)$ is concave.
In order to show strict concavity of $U^{*}(t,\omega, \cdot)$, one can proceed as follows. First, for some positive numbers $x_1\neq x_2$, let $\mathbf c^i= (c^{i,1}, \dots, c^{i,m})$ be such that 
\begin{equation}\label{9172}\begin{array}{rcl}
\sum\limits_{k=1}^mS^k_tc^{i, k}&\leq &x_i,\quad and\\
U^{*}\left(t,\omega,x_i\right) &=& U(t,\omega, c^{i,1},\dots,c^{i,m}),\quad i = 1,2.\\
\end{array}\end{equation}
The existence of such $\mathbf c^i$'s follows from compactness of the domain of the optimization problem in the definition of $U^{*}(t,\omega, x)$ (for every $x>0$) and upper semicontinuity of $U(t,\omega,\cdot)$. Since in \eqref{9172}, $\mathbf c^i$ necessarily satisfies inequality $\sum\limits_{k=1}^mS^k_tc^{i, k}\leq x_i$ with equality, $i=1,2$, from the strict monotonicity of $U(t,\omega,\cdot)$ in every spatial component and $x_1\neq x_2$, we deduce that $\mathbf c^1\neq \mathbf c^2$. Consequently, from {\it strict} concavity of $U(t,\omega,\cdot)$, we get
\begin{displaymath}\begin{array}{rcl}
U^{*}\left(t,\omega,\tfrac{x_1 + x_2}{2}\right) &=& \sup\limits_{\substack{(c_1,\dots,c_m)\in\R^m_{+}:\\\sum\limits_{k = 1}^mc_kS^k_t(\omega) \leq \tfrac{x_1 + x_2}{2}}} U(t,\omega, c_1,\dots,c_m) \\
&\geq & U\left(t,\omega, \tfrac{c^{1,1} +c^{2,1}}{2} ,\dots,\tfrac{c^{1,m}+c^{2,m}}{2}\right)\\
&> & \tfrac{1}{2}U\left(t,\omega,c^{1,1},\dots,c^{1,m}\right) + 
\tfrac{1}{2}U\left(t,\omega, c^{2,1} ,\dots,c^{2,m}\right)\\
&=&\tfrac{1}{2}U^{*}\left(t,\omega,x_1\right) + \tfrac{1}{2}U^{*}\left(t,\omega,x_2\right). \\
\end{array}
\end{displaymath}
Therefore, $U^{*}(t,\omega, \cdot)$ is strictly concave. 
As $U^{*}(t,\omega, \cdot)$  is increasing and strictly concave, it is {\it strictly} increasing.

For every $(t,\omega)\in[0,\infty)\times\Omega$ and $x>0$, using the
Inada conditions for $U(t,\omega, \cdot)$ one can show that there exists $(c_1,\dots,c_m)$ in the {\it interior} of the first orthant, such that $\sum\limits_{i = 1}^mc_iS^i_t(\omega) = x$ and $U^*(t,\omega, x) = U(t,\omega, c_1,\dots,c_m)$. 
As a result, differentiability of $U^{*}(t,\omega, \cdot)$ (in the third argument) follows from differentiability of $U(t,\omega, \cdot)$ and general properties of the subgradient of the image function, see e.g., \cite[Corollary D.4.5.2]{LH04}.
As $U^{*}(t,\omega, \cdot)$ is concave and differentiable, we deduce that $U^{*}(t,\omega, \cdot)$ is {\it continuously} differentiable in the interior of its domain, see \cite[Theorem D.6.2.4]{LH04}.
The Inada conditions for $U^{*}(t,\omega, \cdot)$ follow from the (version of the) Inada conditions for $U(t,\omega, \cdot)$ and \cite[Theorem D.4.5.1, p.192]{LH04}. 

For every $(t,\omega)\in[0,\infty)\times\Omega$, as $U(t,\omega,\cdot)$ is a closed concave function, using e.g., \cite[Theorem 9.2, p. 75]{Rok}, we deduce that $U^{*}(t,\omega,\cdot)$ is also a {\it closed} concave function\footnote{Note that in general, the image of a closed convex or concave function under a linear transformation need not be closed, see a discussion in \cite[p.97]{LH04}.}. In particular, we get $$U^{*}(t,\omega,0)=\lim_{z\downarrow0}U^{*}(t,\omega,z),\quad(t,\omega)\in[0,\infty)\times\Omega.$$

Finally, for every $x\geq 0$, $U^{*}(\cdot, \cdot, x)$ is optional as a supremum of countably many optional processes (where from continuity of $U(t,\omega, \cdot)$ in the relative interior of its effective domain, it is enough to take the supremum (in the definition of $U^*(t,\omega, \cdot)$) over the $m$-dimensional vectors, whose components take only {\it rational} values).

\end{proof}
\begin{rem}
Lemma \ref{lem:U*} asserts that $U^{*}$ satisfies Assumption 2.1 in \cite{Mostovyi2015}.
\end{rem}
 For every $x>0$, we denote by $\A^*(x)$ the set of $1$-dimensional optional processes $c^{*}$, for which there exists an $\R^d$-valued predictable $\widetilde S$-integrable process $H$, such that
\be\nonumber
X_t \triangleq x + \int_0^tH_u\,\ud \widetilde S_u - \int_0^tc^{*}_u\,\ud \kappa_u,
\qquad t\geq0,
\ee
is nonnegative, $\PP$-a.s. We also define
\be	\label{eq:u*}
u^*(x) \triangleq \sup_{c^* \in\A^*(x)}\expec\left[\int_0^{\infty}U(t,\omega,c^{*}_t(\omega))\,\ud\kappa_t(\omega)\right],\quad x>0.
\ee
with the convention analogous to \eqref{9171}:
\be	\nonumber
\expec\left[\int_0^{\infty}U^*(t,\omega,c^{*}_t(\omega))\,\ud\kappa_t(\omega)\right]\triangleq-\infty,\quad if \quad \expec\left[\int_0^{\infty}{U^*}^{-}(t,\omega,c^{*}_t(\omega))\,\ud\kappa_t(\omega)\right] = \infty.
\ee
\begin{proof}[Proof of Theorem \ref{mainTheorem}]
Let $x> 0$ be fixed and $c\in \mathcal A(x)$. Then $c^{*}_t \triangleq \sum\limits_{k= 1}^mc^k_tS^k_t$, $t\geq 0$, is an optional process such that $c^*\in\mathcal A^*(x)$. Therefore, 
\be\label{eq:3211}
u^*(x) \geq u(x) >-\infty,\quad x>0.
\ee 
Since $U^*$ satisfies the assertions of Lemma \ref{lem:U*}, standard techniques in convex analysis show that $-V^{*}$ has the same properties as $U^*$.  
Therefore, optimization problems \eqref{eq:u*} and \eqref{dualProblem} satisfy the assumptions of \cite[Theorem 3.2]{Mostovyi2015}
 . Consequently,  \cite[Theorem 3.2]{Mostovyi2015} applies, which in particular asserts that $u^*$ and $v$ are finite-valued and that for every $x>0$, the exists a strictly positive optional process $\widehat{c^{*}}(x)$, the unique maximizer to \eqref{eq:u*}.

Let us consider 
\be\label{eq:tmp}
\sup\limits_{\substack{(x_1,\dots,x_m)\in\R^m_{+}:\\\sum\limits_{k = 1}^mx_kS^k_t(\omega) \leq \widehat{c^{*}_t}(x)(\omega)}} U\left(t,\omega,x^1,\dots,x^m\right),\quad (t,\omega)\in [0,\infty)\times \Omega,
\ee
and define a correspondence 
$\varphi:[0,\infty)\times \Omega \twoheadrightarrow \R^m$ 
as follows $$\varphi(t,\omega) \triangleq \left\{(x_1,\dots,x_m)\in\R^m_{+}:\sum\limits_{k = 1}^mx_k S^k_t(\omega) \leq\widehat{c^{*}_t}(x)(\omega)\right\}.$$
From {\it strict} positivity of the $S^k$'s and positivity and $(\ud \kappa\times \PP)$-a.e. finiteness of $\widehat {c^{*}}(x)$ (by \cite[Theorem 3.2]{Mostovyi2015}), we deduce that $\varphi$ has nonempty\footnote{Note that the origin in $\R^m$ is in $\varphi(t,\omega)$ for every $(t,\omega)\in[0,\infty)\times \Omega$.} compact values $(\ud \kappa\times \PP)$-a.e. Let us consider the lower inverse of $\varphi^l$ defined by 
$$\varphi^l(G)\triangleq \left\{(t,\omega)\in[0,\infty)\times \Omega:~\varphi(t,\omega)\cap G \neq \emptyset \right\},\quad G\subset{\R^m}.$$
Let us also consider a subset of $\R^m$ of the form $A\triangleq [a_1,b_1]\times \dots \times [a_m,b_m],$ where $a_i$'s and $b_i$'s are real numbers. In view of the weak measurability of $\varphi$  (see \cite[Definition 18.1, p. 592]{AB}) that we are planning to show, it is enough to consider $b_i\geq 0$, $i=1,\dots,m$. In addition, let us set $\bar a_i = \max(0, a_i)$. One can see that for such a set $A$, as
$$\varphi^l(A) = \varphi^l([\bar a_1, b_1]\times \dots\times [\bar a_m, b_m]),$$ 
we have $$\varphi^l(A) = \left\{(t,\omega):~\sum\limits_{i = 1}^m \bar a_iS^i_t(\omega)\leq \widehat{c^{*}_t}(x)(\omega)\right\}.$$
As $\widehat{c^{*}}(x)$ and $S^i$'s are optional processes and since $\varphi^l\left(\bigcup\limits_{n\in\mathbb N}A_n\right) = \bigcup\limits_{n\in\mathbb N}\varphi^l(A_n)$ (see \cite[Section 17.1]{AB}, where $A_n$'s are subsets of $\R^m$), we deduce that $\varphi^l(G)\in\cO$ for every open subset $G$ of $\R^m$,
i.e., $\varphi$ is weakly measurable.
As $U$ is a Carath\'eodory function (see \cite[Definition 4.50, p. 153]{AB}), we conclude from \cite[Theorem 18.19, p. 605]{AB} that there exists an {\it optional} $\R^m$-valued process $\widehat c_t(x)$, $t\in[0,T]$, the maximizer of \eqref{eq:tmp} for $(\ud \kappa\times \PP)$-a.e. $(t,\omega)\in[0,\infty)\times \Omega$. The uniqueness of such a maximizer follows from {\it strict} concavity of $U(t,\omega, \cdot)$ (for every $(t,\omega)\in [0,\infty)\times\Omega$)\footnote{\cite[Theorem 18.19, p. 605]{AB} gives a maximizer, which is a measurable multifunction, and from the uniqueness of the maximizer it is a single valued multifunction, for which the concept of measurability coincides with measurability for functions.}.  As $\widehat{c^{*}}(x)\in\mathcal A^*(x)$, we deduce that $\widehat c(x)\in\A(x)$. Combining this with \eqref{eq:3211}, we conclude that $\widehat c(x)$ is the unique (up to an equivalence class) maximizer to \eqref{primalProblem}. 

For $x>0$, let $\widehat c^i_t(x)$, $i = 1,\dots,m$, denote the components of $\widehat c_t(x)$.
 As $\sum\limits_{i = 1}^m \widehat c^i_t(x)(\omega) S^i_t(\omega) = \widehat c_t^*(\omega)$, $(\ud\kappa\times \mathbb P)$-a.e., (where the argument here is similar to the discussion after \eqref{9172}) relations \eqref{453}, \eqref{eq:3221}, \eqref{451}, and \eqref{452} follow from \cite[Theorem 3.2]{Mostovyi2015}, whereas \eqref{eq:v_defl} results from \cite[Theorem 3.3]{Mostovyi2015} (equivalently, from \cite[Theorem 2.4]{MostovyiNUPBR}). In turn, combining \eqref{eq:3221} with 
\cite[Theorem D.4.5.1]{LH04}, we get
\be\nonumber
\begin{array}{rcl}
\widehat Y_t(\omega) &=& U^*_x\left(t,\omega,\widehat {c_t^*}(x)(\omega)\right) \\
&=& \left\{ s(t,\omega) \in \R:~
S^i_t(\omega)s(t,\omega) = {\partial_{x_i}U}\left(t,\omega, \widehat c^1(x)(\omega),\dots,\widehat c^m(x)(\omega)\right),~i = 1,\dots,m\right\} \\
&&\hspace{114mm} (\ud\kappa\times \mathbb P){\text -a.e.,} \\
\end{array}
\ee
i.e., \eqref{eq:3212} holds.
\end{proof}


%
%

\bibliographystyle{alpha}
\bibliography{lit2}
\end{document}